\documentclass[runningheads]{llncs}

\usepackage{graphicx}

\usepackage{amsmath}
\usepackage{amssymb}
\usepackage{enumitem}
\usepackage{booktabs}
\usepackage{colortbl}
\usepackage{xcolor}
\usepackage[disable]{todonotes}
\definecolor{light-gray}{gray}{0.95}

\usepackage{thmtools}
\usepackage{thm-restate}

\newcommand{\X}{\mathcal{X}}
\newcommand{\Var}{\mathit{Var}}
\newcommand{\var}{\mathit{var}}

\newcommand{\eq}{\,\dot =\,}
\newcommand{\eqbydef}{\triangleq}
\newcommand{\dfness}[3]{\lceil#1\rceil_{#2}^{#3}}

\newcommand{\aaa}[2][]{\todo[author=Andrei,color=green!10!green,#1]{#2}}
\newcommand{\dl}[2][]{\todo[author=Dorel,color=red!10!white,#1]{#2}}

\begin{document}
\title{Unification in Matching Logic}
\subtitle{Extended Version }
%
%
\author{Andrei Arusoaie
\and
Dorel Lucanu
}
\authorrunning{A. Arusoaie, D. Lucanu}
%
\institute{Alexandru Ioan Cuza University, Ia{\c s}i, Romania,
\email{\{arusoaie.andrei,dlucanu\}@info.uaic.ro}}

\maketitle              
\begin{abstract}
Matching Logic is a framework for specifying programming language semantics and reasoning about programs.
Its formulas are called \emph{patterns} and are built with variables, symbols, connectives and quantifiers.
A pattern is a combination of structural components (term patterns), which must be matched, and constraints (predicate patterns), which must be satisfied.
Dealing with more than one structural component in a pattern could be cumbersome because it involves multiple matching operations. A source for getting patterns with many structural components is the conjunction of patterns. 
Here, we propose a method that uses a syntactic unification algorithm to transform conjunctions of structural patterns into equivalent patterns having only one structural component and some additional constraints.
We prove the soundness of our approach, we discuss why the approach is not complete and we provide sound strategies to generate certificates for the equivalences, validated using Coq.
\end{abstract}

\section{Introduction}
\label{sec:introduction}
Matching Logic~\cite{rosu-2017-lmcs} (hereafter shorthanded  as ML) is a novel framework which is currently used for specifying programming languages semantics~\cite{ellison-rosu-2012-popl,DBLP:conf/pldi/HathhornER15,park-stefanescu-rosu-2015-pldi,DBLP:conf/popl/BogdanasR15} and for reasoning about programs~\cite{rosu-stefanescu-2012-fm,stefanescu-park-yuwen-li-rosu-2016-oopsla,stefanescu-ciobaca-mereuta-moore-serbanuta-rosu-2014-rta,DBLP:conf/lics/RosuSCM13,DBLP:conf/wrla/RusuA16,DBLP:conf/birthday/LucanuRAN15,arusoaie:hal-01627517}.
The logic is inspired from the domain of programming language semantics and it aims to use the operational semantics of a programming language as a basis for both \emph{execution} and \emph{verification} of programs.

\aaa{Oare e nevoie sa pastram acest paragraf pentru versiunea de conferinta?}
\dl{Decidem in functie de spatiu.}
On the program verification side, ML has some advantages over the existing program verification logics. The logic is \emph{parametric} in the operational semantics of a language. One can \emph{execute} the semantics against test suites and then use the \emph{same} semantics for verification. Therefore, one can detect issues in the semantics at an early stage and fix them right away, thus, providing additional trust in the semantics. The proof system of ML is proved sound and (relatively) complete for \emph{all} languages, unlike in the existing Floyd-Hoare logics, where the soundness of proof systems needs to be proved separately for each language. Moreover, ML eliminates the need to prove consistency relations between the operational semantics (used for execution) and the axiomatic semantics (used for verification) as it is often the case when using the traditional approaches.

The ML formulas, called \emph{patterns}, are built using variables, symbols, connectives and quantifiers. A pattern is evaluated to the set of values that \emph{matches} it. ML makes no distinction between function symbols and predicate symbols. Not having this distinction increases the expressivity of the language, where various notions (e.g., \emph{function}, \emph{equality}) can be specified using symbols that satisfy some axioms.

An example of such a ML formula is $\varphi_1$ below: it matches over the set of lists that start at address $p + 2$ and store the sequence $a$ which contains an even number on the third position:

$$\varphi_1 \eqbydef \mathtt{list}(p+2, a) \land  \exists k . (\mathit{select}~a~3) = 2 * k$$


\noindent
Basically, the novelty in ML w.r.t. first-order logics is that \emph{structural} components are formulas as well. 
In our example, $\varphi_1$ is a conjunction of a structural component $\mathtt{list}(p+2, a) $ -- that is, a list that starts at address $p+2$ which stores a sequence implemented as an array (encoded using the $\mathit{select}$-$\mathit{store}$ axioms ), -- and a constraint $\exists k . (\mathit{select}~a~3) = 2 * k$.
In ML, the structural components are called \emph{term patterns}, whereas the constraints are called \emph{predicates patterns}.


The conjunction of two ML patterns may produce a new pattern with more than one structural component, as shown here:
$$
\underbrace{\overbrace{\mathtt{list}(p+2, a)}^{\rm structure} \land 
\overbrace{\exists k . (\mathit{select}~a~3) = 2 * k}^{\rm constraint}}_{\varphi_1} \land 
\underbrace{\overbrace{\mathtt{list}(q, (\mathit{store}~b~3~y))}^{\rm structure} \land 
\overbrace{y>2}^{\rm constraint}}_{\varphi_2} 
$$
Finding a set of elements that matches the conjunction $\varphi_1 \land \varphi_2$ is not necessarily an easy task mainly because both structural components $(\mathtt{list}(p+2, a)) $ and $(\mathtt{list}(q, (\mathit{store}~b~3~y)))$ need to be matched simultaneously. 
In theory, this set is the intersection of the sets that match $\varphi_1$ and $\varphi_2$ independently. 

In practice, dealing with multiple structural components in one formula is cumbersome. 
Reasoning with such formulas is a burden for larger formulas.
Also, when mixing multiple structural components in one formula we lose the separation between structure and constraint. 
This separation is essential when implementing a ML prover, where the constraints can be handled separately using existing SMT solvers. In our examples above, the constraints of both $\varphi_1$ and $\varphi_2$ can be dealt with using existing SMT solvers like Z3~\cite{z3} or CVC4~\cite{cvc4} since they provide theories for handling arrays and quantifiers.
A more convenient approach would be to work with formulas that have only one structural component. 

In ML, the semantics of $\varphi_1 \land \varphi_2$ is the largest set of elements matching $\varphi_1$ and $\varphi_2$. Thus, the conjunction of two patterns can be seen as a semantic unification of the two patterns. So, it makes sense to relate syntactic unification to this notion of semantic unification~\cite{rosu-2017-lmcs}.
Let us consider the particular case when $\varphi_i \eqbydef t_i \land \phi_i$, where $t_i$ is a term pattern and $\phi_i$ is a predicate pattern, $i \in \{ 1, 2\}$. In this case:\footnote{For the sake of presentation, we assume here that all patterns have the same sort. Also, the last equality in the sequence $t_1 \land t_2 \land \phi_1 \land \phi_2 = t_1 \land (t_1 = t_2 ) \land \phi_1 \land \phi_2$ holds because of a lemma which is presented in the technical section of the paper.}
$$\varphi_1 \land \varphi_2 = t_1 \land \phi_1 \land t_2 \land \phi_2 = t_1 \land t_2 \land \phi_1 \land \phi_2 = t_1 \land (t_1 = t_2 ) \land \phi_1 \land \phi_2$$
The predicate patterns expressing the equality of two term patterns  $t_1 = t_2$ cannot be handled, e.g.,  by SMT solvers. Therefore, it would be more convenient to reduce it to a simpler equivalent predicate $\phi^{t_1 = t_2}$, which can be handled using external provers. In addition, it would be worth to produce a formal proof of the equivalence between $t_1 = t_2$ and $\phi^{t_1 = t_2}$.


At a first sight, unification of terms seems to be useful here. If $\sigma$ is the most general unifier of $t_1$ and $t_2$, seen as first-order terms, then $t_1\sigma = t_2\sigma$. Unifiers are substitutions, and substitutions can be transformed into ML formulas~\cite{DBLP:journals/cl/ArusoaieLR15}. 

In our list example, $\mathtt{list}(p+2, a)$ and $\mathtt{list}(q, (\mathit{store}~b~3~y))$ have $\sigma = \{ q \mapsto p+2, a \mapsto (\mathit{store}~b~3~y) \}$ as the most general unifier. Translating $\sigma$ to a formula results in $\phi^\sigma \eqbydef (q = p + 2) \land (a = (\mathit{store}~b~3~y))$. For this particular case, the term pattern equality $\mathtt{list}(p+2, a) = \mathtt{list}(q, (\mathit{store}~b~3~y))$ is equivalent to $\phi^\sigma$. Moreover, the semantic unifier $\varphi_1\land \varphi_2$ is also equivalent to $\mathtt{list}(p+2, a) \land \phi^\sigma \land (\exists k . (\mathit{select}~a~3) = 2 * k) \land y > 0$. This form is now convenient since it has only one structural component and a constraint manageable by an SMT solver.

\paragraph{Contributions.} 
We show that $\phi^{t_1 = t_2}$ can be obtained using the most general unifier $\sigma$ of $t_1$ and $t_2$, whenever it exists.
The proof of the equivalence between $t_1 = t_2$ and $\phi^\sigma$ is not trivial and, surprisingly, it depends on the algorithm used to compute the most general unifier. Our proof uses the syntactic unification algorithm proposed by Martelli and Montanari~\cite{martelli}.
Since the equivalence is proved only for the case when the most general unifier exists, we say that this algorithm is \emph{sound} for semantic unification in ML. 

Unfortunately, this algorithm is not \emph{complete} for semantic unification: if the terms $t_1$ and $t_2$ are not syntactically unifiable, then there are no guarantees that $t_1 \land t_2$ is a "contradiction" in ML. We present a detailed analysis of this aspect and we provide a counterexample.

Finally, a \emph{provableness} property of the Martelli-Montanari unification algorithm is shown: we provide a sound strategy to generate a proof certificate of the equivalence between $t_1 \land t_2$ and $t_1 \land \phi^\sigma$ with $\sigma$ the most general unifier of $t_1$ and $t_2$. This proof uses the rules of the ML proof system~\cite{rosu-2017-lmcs}, and the main idea is to transform the steps of the unification algorithm into sequences of proof steps. 
The proposed approach is validated by a Coq encoding, which mechanically checks the correctness of the applied strategy.


All these contributions explicitly establish the relationship between syntactic unification and semantic unification in ML, as summarised be the next table:

\begin{center}
\begin{tabular}{|c|c|c|}
\hline
Syntactic term unification & Semantic unification in ML & Where to find it\\
\hline
unification of $t_1$ and $t_2$ & $t_1 \land t_2$ & - \\
$\textrm{substitution }\sigma$  & $\phi^\sigma$ & Definition~\ref{def:sp}\\
$\sigma(t_1) = \sigma(t_2)$ & $\phi^\sigma \rightarrow t_1 = t_2$ & Lemma~\ref{lem:unif} \\
$\sigma = \mathit{mgu}(t_1, t_2)$ & $t_1 \land t_2 = t_1 \land \phi^\sigma = t_2 \land \phi^\sigma$ & Theorem~\ref{th:main}\\
syntactic unification algorithm & proof certificates & Section~\ref{sec:proofs}\\
\hline
\end{tabular}
\end{center}

\subsubsection{Paper organisation.}
In Section~\ref{sec:sunif} we recall the main notions and notations from the unification theory that we use in this paper. 
Section~\ref{sec:ml} includes a concise presentation of Matching Logic based on~\cite{rosu-2017-lmcs}. 
In Section~\ref{sec:unif} we show how to find the convenient representation of our semantic unifiers using the syntactic unification algorithm. We prove that the unification algorithm is sound for semantic unification and we discuss why this algorithm is not complete for semantic unification.
In Section~\ref{sec:proofs} we describe sound strategies for generating proofs that can be further used to generate proof certificates.

\section{Preliminaries}

\subsection{Syntactic Unification}
\label{sec:sunif}

We recall from~\cite{Baader99unificationtheory} the notions related to unification that we use in this paper.
We also recall the algorithm for finding the \emph{most general unifier} presented in~\cite{martelli}.

\medskip

Let $S$ be a set of sorts.
We consider a (countably) infinite S-indexed set of variables $\Var$ and a \emph{signature}, i.e., a (finite or countably infinite) S-indexed set of function symbols, $\Sigma$. 
By $T_\Sigma$ we denote the algebra of ground terms and by $T_\Sigma(\Var)$ the corresponding term algebra generated by $\Sigma$. 
To keep the presentation simple (as in~\cite{Baader99unificationtheory}) we do not explicitly show the sorts of the terms unless they cannot be inferred from context. This does not restrict in any way the generality and will be handled properly when transferring all these to Matching Logic.

We use the typical conventions and notations.  Letters $x, y, z$ denote variables and $c,f,g$ denote symbols. Terms are either variables or compound terms of the form $f(t_1, \ldots, t_n)$; $f \in \Sigma_{s_1\ldots s_n,s}$ means that $f$ has \emph{arity} $s_1\ldots s_n,s$, that is, for each $i = \overline{1,n}$, the subterm $t_i$ is of sort $s_i$ and the sort of $f(t_1, \ldots, t_n)$ is $s$.
If $n=0$ then $f$ is a constant and the term $f()$ is simply denoted by $f$.
By $\var(t)$ we denote the set of variables occurring in a term $t$. Substitutions are denoted by symbols $\sigma, \eta, \theta$ or directly as a set of bindings $\{ x_1 \mapsto t_1, \ldots , x_n \mapsto t_n\}$. We use $\iota$ to denote the identity substitution. The application of a substitution $\sigma$ to a term $t$ is denoted $t\sigma$\footnote{Although substitutions are defined only over a set of variables, it is well-known that they can be extended to terms. Also, if a substitution $\sigma$ is not defined for a variable, say $x$, then we consider $\sigma(x) = \iota(x) = x$.}. 
The \emph{composition} of substitutions $\sigma$ and $\eta$  is denoted as $\sigma\eta$. If  $\sigma = \{x\mapsto y, y\mapsto y, z \mapsto 4\}$  and $\eta = \{ y \mapsto 3 \}$ then $\sigma\eta= \{x\mapsto y, y\mapsto y, z \mapsto 4\} \eta  = \{x\mapsto (y\eta), y\mapsto (y\eta), z \mapsto (4\eta)\} = \{ x \mapsto 3,  y \mapsto 3,  z \mapsto 4\}$. 
Two substitutions $\sigma$ and $\eta$ are equal, written $\sigma = \eta$, if they are extensionally equal: $x\sigma=x\eta$ for every variable $x$. 
A substitution $\sigma$ is \emph{more general} than a substitution $\eta$, written as $\sigma \leq \eta$, if there is a substitution $\theta$ such that $\sigma\theta=\eta$.

\begin{example}
Let us consider a sort $s$ and a signature $\Sigma$ that includes the symbols $f, g, c$, where $f, g \in \Sigma_{ss,s}$ and $c \in \Sigma_{,s}$. Then  $g(c,c) \in T_\Sigma$ is a ground term, $f(g(x,c), y) \in T_\Sigma(\Var)$ is a term with variables and $\var(f(g(x,c), y)) = \{x, y\} \subseteq \Var$. A substitution $\eta = \{ x \mapsto c , y \mapsto g(x,z) \}$ applied to $f(g(x,a), y)$ produces $\big(f(g(x,c), y)\big) \eta = f(g(c,c), g(x,z))$. If $\sigma = \{ x \mapsto x' , y \mapsto g(x,z) \}$ then $\sigma \leq \eta$, because there is $\theta = \{ x' \mapsto c \} $ such that $\sigma\theta = \eta$.
\end{example}

\begin{definition}[Unifier, Most General Unifier]
\label{def:unifier}
A substitution $\sigma$ is a \emph{unifier} of two terms $t$ and $t'$ if $t\sigma=t'\sigma$. A unifier $\sigma$ is the \emph{most general unifier} (hereafter shorthanded as \emph{mgu}) if for every unifier $\sigma'$ of $t$ and $t'$ we have $\sigma \leq \sigma'$.
\end{definition}

\begin{example}
\label{ex:unif}
If $t \eqbydef f(g(x,c), y)$ and $t' \eqbydef f(z, y')$ are terms then $\sigma = \{ z \mapsto g(x,c), y \mapsto y' \}$ is a unifier of $t$ and $t'$: $t\sigma = f(g(x,c), y') =t'\sigma$.
\end{example}

Whenever there exists a unifier for two given terms we say that the terms are \emph{unifiable}.
It is not always the case that, given two terms, we can find unifiers for them. For example, recall $t$ from Example~\ref{ex:unif} and consider $t'' \eqbydef g(g(x,c), y)$. Then $t$ and $t''$ are not unifiable because it is impossible to find a substitution $\sigma$ such that $t\sigma = t''\sigma$. 
In the particular context of syntactic unification, for every two unifiable terms there exists a \emph{most general unifier}.

\begin{definition}[Unification problem, Solution, Solved form]
\label{def:unifproblem}
An unification problem $P$ is either a set of pairs of terms  $\{ t_1 \eq t'_1, \ldots, t_n \eq t'_n \}$ or a special symbol $\pmb\perp$. 
A substitution $\sigma$ is a \emph{solution} of a unification problem $P = \{ t_1 \eq t'_1, \ldots, t_n \eq t'_n \}$ if $\sigma$ is a unifier of $t_i$ and $t'_i$, for every $i=\overline{1,n}$. A unification problem $P$ is in \emph{solved form} if $P = \pmb\perp $ or $P = \{ x_1 \eq t'_1, \ldots, x_n \eq t'_n \}$ with $x_i \not\in \var(t_j)$ for all $i,j=\overline{1,n}$. 
\end{definition}

\noindent
Let $\mathit{unifiers(P)} = \{ \sigma \mid \sigma \textit{ is a solution of }P \}$ denote the set of solutions of $P$. If $P = \pmb\perp$ then $\mathit{unifiers(P)} = \emptyset$.
Each unification problem $P = \{ x_1 \eq t'_1, \ldots, x_n \eq t'_n \}$ in solved form defines a substitution $\sigma_P=\{ x_1 \mapsto t'_1, \ldots, x_n \mapsto t'_n \}$.

Among the well-known algorithms for finding the most general unifier we encounter the unification by recursive descent~\cite{Robinson:1965:MLB:321250.321253}, and a rule-based approach for finding the mgu~\cite{Herbrand1971,martelli}.
 The latter is presented in Figure~\ref{fig:unif} and it consists of a set of transformation rules of the form $P \Rightarrow P'$ applied over unification problems $P$ and $P'$.
\begin{figure}[htp]
\begin{center}
\begin{tabular}{lcl}
{\bf Delete}: & & $P \cup \{t \eq t\} \Rightarrow P $\\
{\bf Decomposition}: & & $P \cup \{f(t_1, \ldots, t_n) \eq f(t'_1, \ldots, t'_n)\} \Rightarrow P \cup \{t_1 \eq t'_1, \ldots, t_n \eq t'_n\}$\\
{\bf Symbol clash}: && $P \cup \{f(t_1, \ldots, t_n) \eq g(t'_1, \ldots, t'_n)\} \Rightarrow \pmb\perp$\\
{\bf Orient}: & & $P \cup \{f(t_1, \ldots, t_n) \eq x\} \Rightarrow P \cup \{x \eq f(t_1, \ldots, t_n)\}$\\
{\bf Occurs check}: & & $P \cup \{x \eq f(t_1, \ldots, t_n)\} \Rightarrow \pmb\perp$, if $x \in \var(f(t_1, \ldots, t_n))$\\
{\bf Elimination}: && $P \cup \{x \eq t\} \Rightarrow P\{x \mapsto t\} \cup \{ x \eq t \}$ if $x\not\in\var(t), x\in\var(P)$
\end{tabular}
\end{center}
\caption{A rule-based algorithm for syntactic unification}
\label{fig:unif}
\end{figure}%

\begin{remark}
\label{rem:unifalg}
We recall from~\cite{Baader99unificationtheory} the main properties of the unification algorithm in Figure~\ref{fig:unif}. 
\footnote{It is not the purpose of this paper to prove these results. The interested reader is referred to~\cite{Baader99unificationtheory} for complete proofs and details.}
If $P$ as a unification problem then:
\begin{enumerate}
\item\emph{Progress}: If $P$ is not in solved form, then there exists $P'$ such that $P \Rightarrow P'$. 
\item\emph{Solution preservation}: If $P \Rightarrow P'$ then $\mathit{unifiers(P)} = \mathit{unifiers(P')}$.
\item\emph{Termination}: There is no infinite sequence $P \Rightarrow P_1 \Rightarrow P_2 \Rightarrow \cdots$. 
\item\emph{Most general unifier}: If $\theta$ is a solution for P, then for any maximal sequence of transformations $P \Rightarrow^{!} \!P'$ either $P'$ is $\pmb\perp$ or $\sigma_{\!P'} \leq \theta$.
 If there is no solution for $P$ then $P'$ is $\pmb\perp$.
\end{enumerate}
\end{remark}

\noindent
The properties listed in Remark~\ref{rem:unifalg} essentially say that the algorithm in Figure~\ref{fig:unif} produces the most general unifier when it exists.
Note that this algorithm does not impose any strategy to apply the rules.

\begin{example}
Recall $t \eqbydef f(g(x,c), y)$ and $t' \eqbydef f(z, y')$ from Example~\ref{ex:unif}. Consider the unification problem $P = \{ t \eq t' \}$. Using the unification algorithm we obtain:
\begin{center}
\begin{tabular}{llcl}
$P= \{ t \eq t' \} \eqbydef$ &$\{f(g(x,c), y) \eq f(z, y')\}$ & $\Rightarrow$ & ({\bf Decomposition})\\
&$\{g(x,c) \eq z, y \eq  y'\}$ & $\Rightarrow$ & ({\bf Orient})\\
&$\{z \eq g(x,c), y \eq  y'\}$ & $\eqbydef P'$ \\
\end{tabular}
\end{center}
The obtained unification problem $P'$ is in solved form; the corresponding substitution $\sigma_{P'} = \{z \mapsto g(x,c), y \mapsto  y'\} $ is the most general unifier of $t$ and $t'$.
\end{example}

When it exists, the most general unifier is not unique. 
By composition with renaming substitutions we can generate an infinite set of \emph{mgus}. In general, we say that mgus are unique up to a composition with a renaming substitution.

\subsection{Matching Logic}
\label{sec:ml}

Matching Logic~\cite{rosu-2017-lmcs,DBLP:conf/rta/Rosu15} started as a logic over a particular case of constrained terms~\cite{rosu-stefanescu-2012-fm,DBLP:conf/lics/RosuSCM13,stefanescu-ciobaca-mereuta-moore-serbanuta-rosu-2014-rta,DBLP:conf/icse/RosuS11,DBLP:conf/wrla/RusuA16,arusoaie:hal-01627517,DBLP:conf/birthday/LucanuRAN15}, but now it is developed as a solid program logic framework. 
Here we recall from~\cite{rosu-2017-lmcs} the particular definitions and notions of ML that we use in this paper. 
This subsection is longer than an usual one for preliminaries. Since Matching Logic is a quite recent research contribution including new atypical concepts and results, we decided to present it with more details and examples. This makes the paper self-content.

\smallskip

ML formulas are defined over a \emph{many-sorted signature} $(S, \Sigma)$, where $\Sigma$ is a $S^*\times S$-indexed set of \emph{symbols}. The formulas in ML are \emph{patterns}:

\begin{definition}[ML Formula]
\label{def:mlsyntax}
A pattern $\Sigma$-pattern $\varphi_s$ of
sort $s$ is defined by:\\[1ex]
\centerline{
$\varphi_s::=x_s\mid f(\varphi_{s_1},\ldots,\varphi_{s_n})\mid \neg \varphi_s\mid \varphi_s\land \varphi_s\mid \exists x.\varphi_s$
}\\[1ex]
where $x_s$ ranges over the variables of sort $s$ ($x_s\in \X_s \subseteq \Var_s$), $f$ ranges over $\Sigma_{s_1\ldots s_n,s}$, and $x$ ranges over the set of variables (of any sort). 
\end{definition}

\noindent
The derived patterns are defined as expected: $\top\!\!_s\eqbydef \exists x.x$ ($x$ of sort $s$), $\bot_s\eqbydef\neg\top\!\!_s$\footnote{Note that $\bot$ is different from the (bold) symbol $\pmb\perp$ used in Section~\ref{sec:unif}.}, $\varphi_1\lor\varphi_2\eqbydef\neg(\neg\varphi_1\land\lnot\varphi_2)$, $\varphi_1\rightarrow\varphi_2\eqbydef\neg\varphi_1\lor\varphi_2$, $\varphi_1\leftrightarrow\varphi_2\eqbydef(\varphi_1\rightarrow\varphi_2)\land(\varphi_2\rightarrow\varphi_1)$.

\begin{example}
\label{ex:syntax}
Let $\mathit{Nat}$ be a sort and $\Sigma$ a signature which includes symbols $o \in \Sigma_{\mathit{Nat}}$ and $\mathit{succ} \in \Sigma_{\mathit{Nat}, \mathit{Nat}}$. Then, $o$, $\mathit{succ}(o)$, $\mathit{succ}(x)$, $o \land \mathit{succ}(o)$, $\lnot(o \land \mathit{succ}(o))$, $o \lor \exists x. \mathit{succ}(x)$ are all ML patterns. 
\end{example}

When sorts are not relevant or can be inferred from the context we drop the sort subscript ($\varphi_s$ becomes $\varphi$).

\begin{definition}[ML model]
\label{def:mlmodel}
A ML \emph{model} $\Sigma$-model $M$ consist of: 
\begin{itemize}
\item S-sorted sets $\{ M_s \}$ for each $s\in S$, where $M_s$ is the \emph{carrier of sort $s$} of M; 
\item a function $f_M: M_{s_1}\times\cdots\times M_{s_n}\to{\cal P}(M_s)$ (note the use of the powerset ${\cal P}(M_s)$ as the co-domain) for each  symbol $f \in \Sigma_{s_1\ldots s_n,s}$.
\end{itemize}
\end{definition}

\begin{example}
\label{ex:model}
Recall the signature $\Sigma$ from Example~\ref{ex:syntax}. A possible $\Sigma$-model $M$ includes a set $M_\mathit{Nat} = \mathbb{N}$, a constant function $o_M$ which evaluates to the singleton set $\{ 0 \}$, and a function $\mathit{succ}_M : \mathbb{N} \to {\cal P}(\mathbb{N})$ which returns a singleton set containing the successor of the given natural number. Here, the interpretation functions have only singleton sets as results. This is not always the case.
Let us enrich $\Sigma$ with a new symbol ${\leq} \in \Sigma_{\mathit{Nat}\mathit{Nat}, \mathit{Nat}}$.
We can choose the following interpretation $\leq_M$ function for the $\leq$ symbol: ${\leq_M} : \mathbb{N} \times \mathbb{N} \rightarrow {\cal P}(\mathbb{N})$, such that $\leq_M(x,y) = \mathbb{N}$ if $x$ is less or equal than $y$, and $\leq_M(x,y) = \emptyset$ otherwise.
\end{example}

The meaning of patterns is given by using \emph{valuations} $\rho$ as in first-order logic, but the result of the interpretation is a \emph{set} of elements that the pattern ``matches", similar to the worlds in modal logic.

\begin{definition}[M-valuations]
\label{def:valuation}
If $\rho : \X\to M$ is a variable valuation and $\varphi$ a pattern, then the extension of $\rho$ to patterns   $\overline{\rho}(\varphi)$ is inductively defined as follows: 
\begin{enumerate}
\item $\overline{\rho}(x)=\{\rho(x)\}$;
\item $\overline{\rho}(f(\varphi_1,\ldots,\varphi_n))=\bigcup\{f_M(v_1,\ldots,v_n)\mid v_i\in\overline{\rho}(\varphi_i), i=1,\ldots,n\}$;
\item $\overline{\rho}(\neg\varphi)=M_s\setminus\overline{\rho}(\varphi)$, where the sort of $\varphi$ is $s$;
\item $\overline{\rho}(\varphi_1\land\varphi_2)=\overline{\rho}(\varphi_1)\cap\overline{\rho}(\varphi_2)$, where $\varphi_1$ and $\varphi_2$ have the same sort;
\item $\overline{\rho}(\exists x.\varphi)=\bigcup_{v\in M_s}\overline{\rho}[v/x](\varphi)$, where $x\in \X_s$ and ${\overline{\rho}[v/x]}$ is the valuation $\rho'$ s.t. $\rho'(y)=\rho(y)$ for all $y\not=x$, and $\rho'(x)=v$. 
\end{enumerate}
\end{definition}

\noindent
When a functional symbol is a constant $c$ (case $2$ in Def.~\ref{def:valuation}) we let $\overline{\rho}(c) \!=\! c_M$. 
Additional constructs can be handled similarly (e.g. $\overline{\rho}(\varphi_1 \lor \varphi_2) = \overline{\rho}(\varphi_1) \cup \overline{\rho}(\varphi_2)$).

\begin{example}
\label{ex:valuation}
Recall the signature $\Sigma$ from Example~\ref{ex:syntax} and the model $M$ from Example~\ref{ex:model}. Also, consider a valuation $\rho: \X \to M$ such that $\rho(x) = 0$. The pattern $\mathit{succ}(x)$ matches over $\{ 1 \}$ because $\overline{\rho}(\mathit{succ}(x)) = \mathit{succ}_M(\rho(x)) = \mathit{succ}_M(0) = \{ 1 \}$.

An interesting pattern is $o \lor \exists x . \mathit{succ}(x)$ since it matches over the entire set $\mathbb{N}$. Indeed, if we consider any valuation $\varrho: \X \to M_\mathit{Nat}$, then 
$\overline{\varrho}(o \lor \exists x . \mathit{succ}(x)) = \{ 0 \} \cup \bigcup_{n \in \mathbb{N}}\overline{\varrho[n/x]}(\mathit{succ}(x)) =  \{ 0 \} \cup \bigcup_{n \in \mathbb{N}}\mathit{succ}_M(n) = M_\mathit{Nat} = \mathbb{N}$.
\end{example}

A particular type of patterns are M-\emph{predicates}. These are meant to capture the usual meaning of predicates, i.e., patterns that can be either true or false. 
\begin{definition}[M-predicates]
\label{def:mpred}
The pattern $\varphi_s$ is an M-\emph{predicate} iff for any valuation $\rho: \X \to M$, $\overline{\rho}(\varphi_s)$ is either $M_s$ or $\emptyset$. Also, $\varphi_s$ is called a \emph{predicate} iff it is a $M$-predicate in all models M. 
\end{definition}

\begin{example}
\label{ex:predicate}
The pattern $o \lor \exists x . \mathit{succ}(x)$ (from Example~\ref{ex:valuation}) is an $M$-predicate because for all $\varrho: \X \to M$ we have  $\overline{\varrho}(o \lor \exists x . \mathit{succ}(x)) = M_\mathit{Nat} = \mathbb{N}$.

The pattern $o \land \mathit{succ}(o)$ is also an $M$-predicate because $\overline{\varrho}(o \land \mathit{succ}(o)) = \overline{\varrho}(o) \cap \overline{\varrho}(\mathit{succ}(o)) = \{ 0 \} \cap \{ \mathit{succ}_M(0) \} = \emptyset$.
\end{example}

\begin{definition}[Satisfaction relation, validity]
\label{def:sat}
 A model $M$ \emph{satisfies} $\varphi$, written $M\models\varphi_s$, if $M_s=\overline{\rho}(\varphi_s)$ for each variable valuation $\rho$.
A pattern $\varphi$ is \emph{valid} (written $\models \varphi$) iff $M \models \varphi$ for all models M.
\end{definition}

\begin{example}
\label{ex:predicate}
Recall the model $M$ from Example~\ref{ex:model}. $M \models o \lor \exists x . \mathit{succ}(x)$ since, for all $\rho: \X \to M$ we have  $\overline{\rho}(o \lor \exists x . \mathit{succ}(x)) = M_\mathit{Nat}$.
\end{example}


\begin{proposition}[Proposition 2.6 in~\cite{rosu-2017-lmcs}]
\label{prop:derived}
Let $\varphi_1$ and $\varphi_2$ be two ML formulas and M a ML model. Then:
\begin{itemize}
\item $M \models \varphi_1 \rightarrow \varphi_2$ iff $\overline{\rho}(\varphi_1) \subseteq \overline{\rho}(\varphi_2)$ for all $\rho:\Var \to M$.
\item $M \models \varphi_1 \leftrightarrow \varphi_2$ iff $\overline{\rho}(\varphi_1) = \overline{\rho}(\varphi_2)$ for all $\rho:\Var \to M$.
\end{itemize}
\end{proposition}


\begin{definition}[ML specifications]
A \emph{matching logic specification} is a triple $(S, \Sigma, F)$, where $F$ contains $\Sigma$-patterns. The $\Sigma$-patterns in $F$  are \emph{axiom patterns}. We say that $\varphi$ is a \emph{semantical consequence} of $F$, written $F \models \varphi$, iff  $M \models F$ implies $M \models \varphi$, for each $\Sigma$-model $M$.
\end{definition}

An important ingredient of ML is the  \emph{definedness} symbol $\dfness{\_}{s_1}{s_2} \in \Sigma_{s_1, s_2}$, with the following intuitive meaning: if $\varphi$ is matched by \emph{some} values of sort $s_1$ then $\dfness{\varphi}{s_1}{s_2}$ is $\top_{\!\!s_2}$, otherwise it is $\bot_{s_2}$. This interpretation is enforced by including the axiom pattern $\dfness{x}{s_1}{s_2}$ in the set of axioms $F$. This symbol and its associated pattern are used to define:
\begin{itemize}
\item \emph{conjunction of patterns with different sorts}: for instance, if the symbol ${\leq_b} \in \Sigma_{\mathit{Nat}\,\mathit{Nat},\mathit{Bool}}$, then the pattern $x \land o \leq_b x$ is not syntactically correct, because $x$ has sort $\mathit{Nat}$ whereas $ o \leq_b x$ has sort  $\mathit{Bool}$. Using definedness we can now write a syntactically correct formula $x \land \dfness{o \leq_b x}{\mathit{Bool}}{\mathit{Nat}}$;

\item \emph{membership pattern}: $x \in^{s_2}_{s_1} \varphi \eqbydef \dfness{x \land \varphi}{s_1}{s_2}$ with $x\in\Var_{s_1}$, where $x$ is another pattern that evaluates to a single value;

\item \emph{equality pattern}: $\varphi =^{s_2}_{s_1} \varphi' \eqbydef \lnot\dfness{\lnot(\varphi \leftrightarrow \varphi')}{s_1}{s_2}$.
\end{itemize}

In ML there is no distinction between function and predicate symbols. However, there is a way to specify that certain symbols are interpreted as functions. These symbols are called \emph{functional} symbols. 

\begin{definition}[Functional patterns]
\label{def:func}
A pattern $\varphi$ is \emph{functional in a model} M iff $\mid \overline{\rho}(\varphi) \mid = 1$ for any valuation $\rho : \Var \to M$. The pattern $\varphi$ is \emph{functional in } F iff it is functional in all models M such that $M \models F$.
\end{definition}

\begin{remark}
\label{rem:oneelem}
In~\cite{rosu-2017-lmcs} (more precisely, Proposition 5.17 in~\cite{rosu-2017-lmcs}) it is shown that functional patterns are interpreted as total functions in models, and their interpretation contains precisely one element. Moreover, given a ML specification $(S, \Sigma, F)$, a pattern $\varphi$ is functional in all $(S, \Sigma, F)$-models iff $F \models \exists y. (\varphi = y)$.
\end{remark}

\begin{example}
\label{ex:fp}
Recall the $o$ and $\mathit{succ}(x)$ patterns from Example~\ref{ex:syntax}. Both patterns are functional in $M_\mathit{Nat}$ (from Example~\ref{ex:model}) since they are interpreted as functions (i.e., $o_M$ and $\mathit{succ}_M$) which return a singleton set. If we want to have functional interpretation for $o$ and $\it succ$ in all models, then we have to add to $F$ the axioms $\exists y. (c = y)$ and $\exists y. (\mathit{succ}(x) = y)$.
\end{example}

The following technical result was proved in~\cite{rosu-2017-lmcs} and establishes the link between equivalence and equality of functional patterns:

\begin{proposition}[Proposition 5.9 in~\cite{rosu-2017-lmcs}]
\label{prop:eq}
If $\varphi$, $\varphi'$ are patterns of sort $s_1$~then:
\begin{itemize}
\item $\overline{\rho}(\varphi =^{s_2}_{s_1} \varphi') = \emptyset$ iff $\overline{\rho}(\varphi) \neq \overline{\rho}(\varphi')$, for any $\rho:\Var\to M$.
\item $\overline{\rho}(\varphi =^{s_2}_{s_1} \varphi') = M_{s_2}$ iff $\overline{\rho}(\varphi) = \overline{\rho}(\varphi')$, for any $\rho:\Var\to M$.
\item $M \models \varphi =^{s_2}_{s_1} \varphi'$ iff $M \models \varphi \leftrightarrow \varphi'$, for any model M.
\item $\models \varphi =^{s_2}_{s_1} \varphi'$ iff $\models \varphi \leftrightarrow \varphi'$.
\end{itemize}
\end{proposition}

\noindent
It is worth noting that the Proposition~\ref{prop:eq} holds only for functional patterns.
When functional patterns have the same sort, the proposition below holds:
\begin{proposition}[Proposition 5.24 in~\cite{rosu-2017-lmcs}]
\label{prop:fpatt}
If $\varphi$ and $\varphi'$ are two functional patterns of the same sort then $\models (\varphi \land \varphi') = \varphi \land (\varphi = \varphi')$.
\end{proposition}

\begin{definition}[Term patterns]
\label{def:terms}
If $f \in \Sigma_{s_1, \ldots, s_n,s}$ is a symbol such that $F$ contains the pattern $ \exists y. (f(x_1, \ldots, x_n) = y)$ then $f$ is a \emph{functional symbol}. \emph{Term patterns} are formulas containing only functional symbols.
\end{definition}

\begin{example}
\label{ex:termpattern}
If $1, f, g$ are symbols in $\Sigma$ and $x$ and $z$ are variables in $\Var$, then 
$t\eqbydef f(x, g(1), g(z))$ is a term pattern if $\exists y.f(x_1, x_2) = y$, $\exists y.g(x_1) = y$, and $\exists y.1 = y$ are semantical consequences of the axioms $F$.
\end{example}

\subsubsection{Substitution.} Sometimes we need to use substitution over ML patterns directly. We use $\varphi[\varphi'/x]$ to denote the formula obtained by substituting $\varphi'$ for variable $x$ in $\varphi$ (we assume $\varphi'$ and $x$ have the same sort):
\begin{enumerate}
\item $x[\varphi'/x] = \varphi'$; ~~$y[\varphi'/x] = y$ when $x \neq y$.
\item $f(\varphi_{1}, \ldots, \varphi_{n})[\varphi'/x] =  f(\varphi_{1}[\varphi'/x], \ldots, \varphi_{n}[\varphi'/x])$
\item $(\lnot \varphi)[\varphi'/x] = \lnot (\varphi[\varphi'/x])$
\item $(\varphi_1 \land \varphi_2)[\varphi'/x] = \varphi_1[\varphi'/x] \land \varphi_2[\varphi'/x]$
\item $(\exists y . \varphi) [\varphi'/x] = \exists y . \varphi[\varphi'/x]$, if $y \not\in \var(\varphi')$; otherwise, a renaming is required.
\end{enumerate}

Our main result use the following technical lemma. For the particular case when the equivalence and the equality are the same it is a consequence of Proposition 5.10 from~\cite{rosu-2017-lmcs}. We include its proof here as an example of Matching Logic reasoning.
\begin{lemma}
\label{lem:substpattern}
If $\varphi$ is a pattern, $t$ is a term pattern, and $x$ is a variable such that $F\models x = t$, then $F\models \varphi[t/x] = \varphi$.
\end{lemma}

\begin{proof}
By induction on $\varphi$, we show: for all $M$ and $\rho:\X \to M$, $\overline{\rho}(\varphi[t/x]) = \overline{\rho}(\varphi)$:
\begin{itemize}
\item $\overline{\rho}(x[t/x]) = \rho(t) = \rho(x)$, which (by Proposition~\ref{prop:eq}) holds since $x = t$;
\item $\overline{\rho}(y[t/x]) = \rho(y)$ when $x \neq y$;
\item $\overline{\rho}(f(\varphi_{1}, \ldots, \varphi_{n})[t/x]) = \overline{\rho}(f(\varphi_{1}[t/x], \ldots, \varphi_{n}[t/x]))$ which, by Definition~\ref{def:mlsyntax} is $\bigcup_{}{\{f_M(v_1, \ldots, v_n) \mid v_i \in \overline{\rho}(\varphi_i[t/x]), i=\overline{1,n}\}}$. Here, we use the inductive hypothesis which says that $ \overline{\rho}(\varphi_{i}) =  \overline{\rho}(\varphi_{i}[t/x])$ for all $i=\overline{1,n}$, and we obtain   
$\bigcup_{}{\{f_M(v_1, \ldots, v_n) \mid v_i \in \overline{\rho}(\varphi_i), i=\overline{1,n}\}} = 
\overline{\rho}(f(\varphi_{1}, \ldots, \varphi_{n}))$;
\item $\overline{\rho}((\lnot \varphi')[t/x]) = \overline{\rho}(\lnot (\varphi'[t/x])) = M \setminus \overline{\rho}(\varphi'[t/x]) = M \setminus \overline{\rho}(\varphi') = \overline{\rho}(\lnot \varphi')$ using $ \overline{\rho}(\varphi'[t/x]) =  \overline{\rho}(\varphi')$ from the inductive hypothesis;
\item $\overline{\rho}((\varphi_1 \land \varphi_2)[t/x]) = \overline{\rho}(\varphi_1[t/x] \land \varphi_2[t/x]) =  \overline{\rho}(\varphi_1[t/x]) \cap \overline{\rho}(\varphi_2[t/x])  = \overline{\rho}(\varphi_1) \cap \overline{\rho}(\varphi_2) = \overline{\rho}(\varphi_1 \land \varphi_2)$ by the inductive hypothesis: $\overline{\rho}(\varphi_i[t/x]) = \overline{\rho}(\varphi_i)$, $i \in \{1,2\}$;
\item $\overline{\rho}((\exists y . \varphi')[t/x]) = \bigcup_{v \in M}\{\overline{\rho}[v/y](\varphi'[t/x])\} = \bigcup_{v \in M}\{\overline{\rho}[v/y](\varphi')\} = \overline{\rho}(\exists y . \varphi')$, with $y \not\in \var(t)$ and $\overline{\rho}[v/y](\varphi'[t/x]) = \overline{\rho}[v/y](\varphi')$  - the inductive hypothesis.
\end{itemize}
Since for all $M$ and $\rho:\X \to M$, $\overline{\rho}(\varphi[t/x]) \!=\! \overline{\rho}(\varphi)$ (by Prop.~\ref{prop:eq}) we have $\varphi[t/x] = \varphi$. \qed
\end{proof}

\subsubsection{The proof system of Matching Logic.}

\newcounter{lprop}
\newcommand{\proprule}[1]{\refstepcounter{lprop}\label{#1}$P$\arabic{lprop}.~}
\renewcommand{\thelprop}{$P$\arabic{lprop}}
\begin{figure}[htbp]
\begin{center}
\begin{tabular}{|l l|}
\hline

\proprule{l1} & $\vdash$ $\varphi \rightarrow (\varphi' \rightarrow \varphi)$\\
\proprule{l2} & $\vdash$ $(\varphi \rightarrow (\varphi' \rightarrow \varphi'')) \rightarrow ((\varphi \rightarrow \varphi') \rightarrow (\varphi \rightarrow\varphi''))$\\
\proprule{l3} & $\vdash$ $(\lnot \varphi' \rightarrow \lnot \varphi) \rightarrow (\varphi \rightarrow \varphi')$\\
\hline
\end{tabular}
\end{center}
\caption{Rules for propositional reasoning in ML.}
\label{fig:prop}
\end{figure}

Matching Logic provides a proof system that is sound and complete (Figure~\ref{fig:proofsystem}). 
The notation $\varphi[\varphi'/x]$ denotes the pattern obtained from $\varphi$ by replacing all free occurrences of $x$ with $\varphi'$.
Note that the propositional calculus reasoning is subsumed by rules \ref{propositional}-\ref{modusponens} of the proof system. 
According to~\cite{koredocs}, \ref{propositional} is in fact a set of rules that includes a version of the implicational propositional calculus (proposed by \L{}ukasievicz~\cite{lukasievicz}) shown in Fig.~\ref{fig:prop}.

\newcounter{original}
\newcommand{\originalrule}[1]{\refstepcounter{original}\label{#1}$R$\arabic{original}.~}
\renewcommand{\theoriginal}{$R$\arabic{original}}

\begin{figure}[htbp]
\begin{center}
\begin{tabular}{|l l|}
\hline
\originalrule{propositional} & $\vdash$ propositional tautologies\\
\originalrule{modusponens} & Modus ponens: $\vdash \varphi_1$ and $\vdash \varphi_1 \rightarrow \varphi_2$ imply $\vdash \varphi_2$\\
\originalrule{free} & $\vdash(\forall x . \varphi_1 \rightarrow \varphi_2) \rightarrow \varphi_1 \rightarrow (\forall x . \varphi_2)$, when $x$ does not occur free in $\varphi_1$\\
\originalrule{universalgeneralization} & Universal generalization: $\vdash \varphi$ implies $\vdash \forall x . \varphi$\\
\originalrule{functionalsubstitution} & Functional substitution: $\vdash (\forall x . \varphi) \land (\exists y . \varphi' = y) \rightarrow \varphi[\varphi'/x]$\\
$R$5.' & Functional variable: $\vdash \exists y . x = y$\\
\originalrule{equalityintroduction} & Equality introduction: $\vdash \varphi = \varphi$\\
\originalrule{equalityelimination} & Equality elimination: $\vdash \varphi_1 = \varphi_2 \land \varphi[\varphi_1/x] \rightarrow \varphi[\varphi_2/x]$\\
\originalrule{membershipequiv} & $\vdash \forall x . x \in \varphi$ iff $\vdash \varphi$\\
\originalrule{membershipequality} & $\vdash x \in y = (x = y)$ when $x, y \in \Var$\\
\originalrule{membershipnegation} &  $\vdash x \in \lnot \varphi = \lnot (x \in \varphi) $\\
\originalrule{membershipconjunction} & $\vdash x \in \varphi_1 \land \varphi_2 = (x \in \varphi_1) \land (x \in \varphi_2)$\\
\originalrule{membershipexists} & $\vdash (x \in \exists y.\varphi) = \exists y . (x \in \varphi)$, with $x$ and $y$ distinct\\
\originalrule{membership} & $\vdash x \!\in\! f(\varphi_1,.., \varphi_{i-1}, \varphi_i, \varphi_{i+1},.., \varphi_n) = \exists y . (y \!\in\! \varphi_i \land f(\varphi_1,.., \varphi_{i-1}, y, \varphi_{i+1},..,\varphi_n))$\\
\hline
\end{tabular}
\end{center}
\caption{Sound and complete proof system of Matching Logic~\cite{rosu-2017-lmcs}}
\label{fig:proofsystem}
\end{figure}

\subsubsection{Unification in Matching Logic.}
\label{sec:uniftoml}

In~\cite{rosu-2017-lmcs}, unification has a semantical definition. 
More precisely, it is defined in terms of conjunctions of patterns. 
In order to explain this better, let us consider two ML patterns: $\varphi$ and $\varphi'$. 
Both patterns can be matched by (possibly infinite) sets of elements, say $\overline{\rho}(\varphi)$ and $\overline{\rho}(\varphi')$, given some variable valuation $\rho$. 
In this context, finding a unifier is the same as finding a pattern $\varphi_u$  that matches over a set of elements included in both $\overline{\rho}(\varphi)$ and $\overline{\rho}(\varphi')$, that is,  $\overline{\rho}(\varphi_u) \subseteq \overline{\rho}(\varphi)\cap\overline{\rho}(\varphi')$, for any $\rho$.
The \emph{most general} pattern $\varphi_u$ that corresponds to the \emph{largest} set 
with this property (i.e.,  $\overline{\rho}(\varphi)\cap\overline{\rho}(\varphi')$), is (by Definition~\ref{def:valuation}) the pattern $\varphi\land \varphi'$.

\section{From Unification Theory to Matching Logic}
\label{sec:unif}
This section is concerned with finding, for two given term patterns $t_1$ and $t_2$, a pattern of the form $t \land \phi$ and having the following properties: 1) $t$ is a term pattern, 2) $\phi$ is a predicate pattern that captures the idea of the most general unifier of $t_1$ and $t_2$, and 3) $\models t_1\land t_2=t\land\phi$.
This particular form ($t \land \phi$) has some very practical advantages compared to $t_1 \land t_2$. 
First, there only one structural part of the formula held by $t$ which is separated from the constraints $\phi$. 
Second, as we show in this section, having a single structural component in a formula allows implementations to reuse existing work on  unification.
Finally, the separation of constraints $\phi$ enables the use of SMT solvers for reasoning.

The idea of transforming the pattern $t_1 \land t_2$ into an equivalent one $t \land \phi$ was suggessted in~\cite{rosu-2017-lmcs}, using an example. Here we propose a general solution that involves the unification algorithm shown in Figure~\ref{fig:unif}. 
Example~\ref{ex:mgu} illustrates how the rules of the unification algorithm are simulated by pattern transformations. Except the step (\ref{a:fpatt}) - which is a direct consequence of Proposition~\ref{prop:fpatt} applied to~(\ref{a:start}) - the rest of the equations correspond to the steps of the algorithm: {\bf Decomposition} for (\ref{a:dec1},\ref{a:dec2},\ref{a:dec3}), {\bf Orient} for (\ref{a:orient}), and {\bf Elimination} for (\ref{a:elim1},\ref{a:elim2}).

\begin{example}
\label{ex:mgu}
$t_1 \land t_2$ can be transformed into an equivalent formula $t \land \phi$:
\begin{align}
t_1 \land t_2 &= f(x, g(1), g(z)) \land f(g(y), g(y), g(g(x)) \label{a:start} \\
&= f(x, g(1), g(z)) \land (f(x, g(1), g(z)) = f(g(y), g(y), g(g(x))) \label{a:fpatt} \\
&= f(x, g(1), g(z)) \land (x = g(y)) \land (g(1) = g(y)) \land (g(z) = g(g(x))) \label{a:dec1}\\
&= f(x, g(1), g(z)) \land (x = g(y)) \land (1 = g(y)) \land (g(z) = g(g(x))) \label{a:dec2}\\
&= f(x, g(1), g(z)) \land (x = g(y)) \land (1 = y) \land (z = g(x)) \label{a:dec3}\\
&= f(x, g(1), g(z)) \land (x = g(y)) \land (y = 1) \land (z = g(x)) \label{a:orient}\\
&= f(x, g(1), g(z)) \land (x = g(1)) \land (y = 1) \land (z = g(x)) \label{a:elim1}\\
&= \underbrace{f(x, g(1), g(z))}_t \land \underbrace{(x = g(1)) \land (y = 1) \land (z = g(g(1)))}_\phi \label{a:elim2}
\end{align}
\end{example}

Mainly, the idea illustrated in Example~\ref{ex:mgu} is to use the syntax unification algorithm to determine $\phi$.
In the rest of this section we introduce several notions and prove some intermediate technical results that we use to formally prove the equality (in the ML sense) of $t \land \phi$ and $t_1 \land t_2$.

\subsection{Encoding unification in ML}
\label{sec:results}
Terms can be naturally expressed in ML  as term patterns provided the following set of axiom patterns which we consider implicitly included in $(S, \Sigma, F)$:
\begin{enumerate}
\item The definedness patterns, needed to define equality and membership;
\item Axioms ensuring that the structural patterns are built only with functional symbols (cf. Definition~\ref{def:terms});
\item Axioms ensuring that all functional symbols used in structural patterns are interpreted as injections:
\begin{equation}
\label{eq:inj}
f(x_1,\ldots,x_n) = f(y_1, \ldots, y_n) \rightarrow x_1 = y_1 \land \ldots \land x_n = y_n.
\end{equation}
\end{enumerate}

\dl[inline]{Partea care urmeaza nu va fi inclusa in submisia la conferinta.}
\aaa[inline]{Agreed.}
Let us explain here why this particular axiom ensures injectivity.
Indeed, in any model $M$ and for any $\rho: \X \to M$, $\overline{\rho}(f(x_1,\ldots,x_n) = f(y_1, \ldots, y_n) \rightarrow x_1 = y_1 \land \ldots \land x_n = y_n) = M$ (cf. Prop.~\ref{prop:derived}) implies that  $\overline{\rho}(f(x_1,\ldots,x_n) = f(y_1, \ldots, y_n)) \subseteq \overline{\rho}(x_1 = y_1 \land \ldots \land x_n = y_n))$. 
If $\overline{\rho}(f(x_1,\ldots,x_n) = f(y_1, \ldots, y_n)) = \emptyset$ (recall that equality is a predicate), then (cf. Prop.~\ref{prop:eq}) $\overline{\rho}(f(x_1,\ldots,x_n)) \neq \overline{\rho}(f(y_1,\ldots,y_n))$, which implies that $f_M(v_1, \ldots, v_n) \neq f_M(u_1, \ldots, u_n)$, for all $v_i  = \rho(x_i)$, $u_i  = \rho(y_i)$, $i=\overline{1,n}$. 
Therefore, for all $v_i$, $u_i$, $v_i \neq u_i$ implies $f_M(v_1, \ldots, v_n) \neq f_M(u_1, \ldots, u_n)$ holds as well, and hence the function $f_M$ is injective.
On the other hand, if $\overline{\rho}(f(x_1,\ldots,x_n) = f(y_1, \ldots, y_n)) = M$, then $\overline{\rho}(f(x_1,\ldots,x_n)) = \overline{\rho}(f(y_1,\ldots,y_n))$ and $\overline{\rho}(x_1 = y_1 \land \ldots \land x_n = y_n) = M$. Therefore,  $f_M(v_1, \ldots, v_n) = f_M(u_1, \ldots, u_n)$, for all $v_i \in \{\rho(x_i)\}$, $u_i \in \{\rho(y_i)\}$, $i=\overline{1,n}$ implies $\bigwedge_i v_i = u_i$ (because, in particular $\overline{\rho}(x_i = y_i)= M$), and $f_M$ is injective.

\medskip
As suggested by Example~\ref{ex:mgu}, our solution for finding an equivalent form $t \land \phi$ for $t_1 \land t_2$ requires the simulation of the unification algorithm shown in Figure~\ref{fig:unif} in ML.  
First, we have to encode unification problems as ML formulas:

\begin{definition}
\label{def:up}
For each unification problem $P = \{ v_1 \eq u_1, \ldots, v_n \eq u_n \}$ we define a corresponding ML predicate $\phi^P \eqbydef \bigwedge_{i = 1}^nv_i = u_i$. Also, $\phi^{\pmb\perp} = \bot$.
\end{definition}

A unification problem in solved form has a corresponding substitution.
These substitutions can be encoded as ML predicates called \emph{substitution patterns}:
\begin{definition}
\label{def:sp}
A \emph{substitution pattern} that corresponds to a substitution $\sigma = \{ x_i \mapsto u_i \mid i=1,\ldots,n \}$ is a predicate of the form $\phi^\sigma\eqbydef\bigwedge_{i=1}^m x_i=u_i$.
\end{definition}

For the particular case when $\sigma$ corresponds to a unification problem $P$ in solved form we have $\phi^{\sigma}=\phi^P$.
For a term pattern $t$, we use the same notation $t\sigma$ to denote the corresponding term pattern obtained after applying substitution~$\sigma$ to $t$, as follows:
 $x_i\sigma=u_i$ if $( x_i \mapsto u_i ) \in \sigma$; $x\sigma=x$ if $(x \mapsto \_) \not\in \sigma$; finally, $f(t_1,\ldots,t_n)\sigma=f(t_1\sigma,\ldots,t_n\sigma)$.

\begin{example}
\label{ex:terms}
Terms $t_1\eqbydef f(x, g(1), g(z))$ and $t_2\eqbydef f(g(y), g(y), g(g(x))$ are term patterns in ML. 
For the unifier $\sigma = \{ x \mapsto g(1), y\mapsto 1, z \mapsto g(g(1))\}$ of $t_1$ and $t_2$ the corresponding substitution pattern $\phi^\sigma$ is $x = g(1) \land y = 1 \land z = g(g(1))$. 
Now, both $t_1\sigma$ and $t_2\sigma$ are the same with the ML pattern $f(g(1), g(1), g(g(1)))$. 
\end{example}

One may be tempted to say that for every term $t$, $t\sigma$ is equal to $t \land \phi^\sigma$.
However, this is not always true. 
For instance, if $t$ is a variable $x\in\X$ such that $x\sigma=x$ and $\overline{\rho}(\phi^\sigma) = \emptyset$, and $\rho : \X \to M$ is a valuation, then $\overline{\rho}(x \land \phi^\sigma) =\rho(x) \cap \overline{\rho}(\phi^\sigma) = \emptyset \not= \{\rho(x)\} = \overline{\rho}(x) = \overline{\rho}(x\sigma)$.
The following lemma formalises the precise relation between $t\sigma$ and $t \land \phi^\sigma$:
\begin{lemma}
\label{lem:subst}
If $t$ is a term pattern and $\sigma$ a substitution then $F \models t\sigma\land\phi^\sigma \leftrightarrow t\land\phi^\sigma$.
\end{lemma}

\begin{proof}
Let us choose an arbitrary model $M$. We have to prove that $M \models t\sigma\land\phi^\sigma \leftrightarrow t\land\phi^\sigma$. 
By Proposition~\ref{prop:eq}, $M \models t\sigma\land\phi^\sigma \leftrightarrow t\land\phi^\sigma$ iff $\overline{\rho}(t\sigma\land\phi^\sigma) = \overline{\rho}(t\land\phi^\sigma)$ for any $\rho: \Var \to M$. By Proposition~\ref{prop:derived}, this holds iff $\overline{\rho}(t) \cap \overline{\rho}(\phi^\sigma) = \overline{\rho}(t\sigma) \cap \overline{\rho}(\phi^\sigma)$. If $\overline{\rho}(\phi^\sigma) = \emptyset$ then this equality holds trivially. 
If $\overline{\rho}(\phi^\sigma) = M$\footnote{This should be $M_s$ where $s\in S$ is the sort of $t\sigma$, but we choose not to show the sort explicitly.}
then 
$\overline{\rho}(t\sigma) \cap M = \overline{\rho}(t) \cap M$  iff $\overline{\rho}(t\sigma) = \overline{\rho}(t)$.
We proceed by structural induction on $t$:

\begin{itemize}
\item \emph{Base case.} $t = x$. Recall that $\phi^\sigma\eqbydef\bigwedge_{i=1}^m x_i=u_i$. We have two sub-cases:
	\begin{enumerate}
	\item $x\in\{x_1,\ldots,x_m\}$: since $\overline{\rho}(\phi^\sigma) = M$ then $\overline{\rho}(\bigwedge_{i=1}^m x_i=u_i) = M$ which implies 
$\bigcap_i^m\overline{\rho}(x_i=u_i) = M$. Thus, $\overline{\rho}(x_i=u_i) = M$ iff  $\overline{\rho}(x_i) =\overline{\rho}(u_i)$ (using Proposition~\ref{prop:eq}) and in particular $\rho(x)=\overline{\rho}(u) = \overline{\rho}(x\sigma)$.

	\item $x\not\in\{x_1,\ldots,x_m\}$: in this case $x\sigma=x$ and $\overline{\rho}(x\sigma) = \rho(x)$.
	\end{enumerate}
\item \emph{Inductive step.} $t = f(t_1, \ldots, t_n)$ and the inductive hypothesis holds for all subterm patterns $t_1, \ldots, t_n$.
Then: $\overline{\rho}(f(t_1, \ldots, t_n)\sigma)= $
$~\overline{\rho}(f(t_1\sigma, \ldots, t_n\sigma))=$
$~ f_M(\overline{\rho}(t_1\sigma), \ldots, \overline{\rho}(t_n\sigma))) =$
$~f_M(\overline{\rho}(t_1), \ldots, \overline{\rho}(t_n)))  =$
$~\overline{\rho}(f(t_1, \ldots, t_n))$, using the inductive hypothesis and Definitions~\ref{def:func} and~\ref{def:valuation}. \qed
\end{itemize}
\end{proof}

Lemma~\ref{lem:sim} shows that the steps performed by the unification algorithm (Figure~\ref{fig:unif}) can be encoded as implications in ML. Note that we only consider the case when the most general unifier exists.

\begin{lemma}
\label{lem:sim}
In the context of Figure~\ref{fig:unif}, if $P \Rightarrow P'$ and $P'\neq\pmb{\bot}$ then $F \models \phi^P \rightarrow \phi^{P'}$, for all unification problems $P$ and $P'$.
\end{lemma}
\begin{proof}
We have to prove that for all models $M$ and for all valuations $\rho:\X \to M$, $\overline{\rho}(\phi^P \rightarrow \phi^{P'}) = M$, that is (by Proposition~\ref{prop:derived}), $\overline{\rho}(\phi^P) \subseteq \overline{\rho}(\phi^{P'})$. Since $\phi^P$ is a predicate, $\overline{\rho}(\phi^P)$ is either $\emptyset$ (in this case the lemma holds trivially) or $M$. Let $\overline{\rho}(\phi^P)= M$; we proceed by case analysis on the rule applied for step $P \Rightarrow P'$:
\begin{enumerate}
\item \label{e:del} {\bf Delete}: $\overline{\rho}(\phi^{P \cup \{t \eq t\}}) = \overline{\rho}(\phi^{P} \land t = t) = \overline{\rho}(\phi^{P}) \cap \overline{\rho}(t = t) \subseteq \overline{\rho}(\phi^{P})$.

\item \label{e:dec} {\bf Decomposition}: On the one hand we have $\overline{\rho}(\phi^{P \cup \{f(t_1, \ldots, t_n) \eq f(t'_1, \ldots, t'_n)\}}) = \overline{\rho}(\phi^{P} \land f(t_1, \ldots, t_n) = f(t'_1, \ldots, t'_n)) = \overline{\rho}(\phi^{P}) \cap \overline{\rho}(f(t_1, \ldots, t_n) \!=\! f(t'_1, \ldots, t'_n))$. 
On the other hand, $\overline{\rho}(\phi^{P \cup \{t_1 \eq t'_1, \ldots, t_n \eq t'_n\}}) = \overline{\rho}(\phi^{P} \land t_1 = t'_1 \land \ldots \land t_n = t'_n) = \overline{\rho}(\phi^{P}) \cap \overline{\rho}(t_1 = t'_1) \cap \ldots \cap \overline{\rho}(t_n = t'_n)$.

If $\overline{\rho}(f(t_1, \ldots, t_n) = f(t'_1, \ldots, t'_n)) = \emptyset$ then the inclusion holds trivially. 

If $\overline{\rho}(f(t_1, \ldots, t_n) \!=\! f(t'_1, \ldots, t'_n)) \!=\! M$, then $\overline{\rho}(f(t_1, \ldots, t_n)) \!=\! \overline{\rho}(f(t'_1, \ldots, t'_n))$ (by Proposition~\ref{prop:eq}) iff $f_M(\overline{\rho}(t_1), \ldots, \overline{\rho}(t_n)) = f_M(\overline{\rho}(t'_1), \ldots, \overline{\rho}(t'_n))$ (recall that $f$ is a functional symbol and $f_M$ is an injective function) which implies $\overline{\rho}(t_1) = \overline{\rho}(t'_1)$, \ldots, $\overline{\rho}(t_n) = \overline{\rho}(t'_n)$. 
Thus, $\overline{\rho}(t_1 = t'_1) = \cdots = \overline{\rho}(t_n = t'_n) = M$.


\item \label{e:ori} {\bf Orient}: $\overline{\rho}(\phi^{P \cup \{f(t_1, \ldots, t_n) \eq x\}}) = \overline{\rho}(\phi^{P} \land f(t_1, \ldots, t_n) = x) = \overline{\rho}(\phi^{P} \land x = f(t_1, \ldots, t_n)) = \overline{\rho}(\phi^{P \cup \{x \eq f(t_1, \ldots, t_n)\}})$.

\item \label{e:elim} {\bf Elimination}: we have $x \not\in \var(t)$, $x \in \var(P)$and we have to show that $\overline{\rho}(\phi^{P \cup \{x \eq t \}}) = \overline{\rho}(\phi^{P}) \cap \overline{\rho}(x = t)$ is included in the set $\overline{\rho}(\phi^{P\{x \mapsto t\} \cup \{ x \eq t \}}) = \overline{\rho}(\phi^{P\{x \mapsto t\}}) \cap \overline{\rho}(x = t) =  \overline{\rho}(\phi^{P}[t/x]) \cap \overline{\rho}(x = t)$. 
If $\overline{\rho}(x = t) = \emptyset$ the inclusion holds trivially. The interesting case is when $\overline{\rho}(x = t) = M$, i.e., $x =t$. In this case it is sufficient to prove $\overline{\rho}(\phi^{P}) = \overline{\rho}(\phi^{P}[t/x])$ which follows from Lemma~\ref{lem:substpattern}.
\item {\bf Occurs check} and {\bf Symbol clash} cannot be applied because $P' \neq \pmb{\bot}$.
 \qed
\end{enumerate}
\end{proof}

\begin{lemma}
\label{lem:unifback}
If $\sigma$ is the most general unifier of $t_1$ and $t_2$ then $F \models (t_1 = t_2) \rightarrow \phi^\sigma$.
\end{lemma}

\begin{proof}
Note that $\sigma$ is obtained using the unification algorithm in Figure~\ref{fig:unif}. The algorithm generates a finite sequence  $\{ t_1 \eq t_2 \} \Rightarrow \cdots \Rightarrow P^\sigma$ where $P^\sigma$ is in solved form and corresponds to mgu $\sigma$. If we apply Lemma~\ref{lem:sim} for each step in this sequence we also have a sequence of valid implications $(t_1 = t_2) \rightarrow \cdots \rightarrow \phi^\sigma$.

\end{proof}

The reversed implication is given by the following lemma:

\begin{lemma}
\label{lem:unif}
If $\sigma$ is a unifier of term patterns $t_1$ and $t_2$ then $\models \phi^\sigma\rightarrow (t_1 = t_2)$.
\end{lemma}

\begin{proof}
We have to prove that for all models $M$ and for all valuations $\rho:\X \to M$, $\overline{\rho}(\phi^\sigma\rightarrow (t_1 = t_2)) = M$. 
By Proposition~\ref{prop:derived}, we have to prove that $\overline{\rho}(\phi^\sigma) \subseteq \overline{\rho}(t_1 = t_2))$. 
The case $\overline{\rho}(\phi^\sigma) = \emptyset$ is trivial.
When $\overline{\rho}(\phi^\sigma) = M$ it is sufficient to prove that $\overline{\rho}(t_1 = t_2) = M$, namely, $\overline{\rho}(t_1) = \overline{\rho}(t_2)$. 

From Lemma~\ref{lem:subst} we have $F\models t_1\sigma\land\phi^\sigma \leftrightarrow t_1\land\phi^\sigma$ and $F \models t_2\sigma\land\phi^\sigma \leftrightarrow t_2\land\phi^\sigma$.
This implies that $\overline{\rho}(t_1\sigma) \cap \overline{\rho}(\phi^\sigma) = \overline{\rho}(t_1) \cap \overline{\rho}(\phi^\sigma)$ and $\overline{\rho}(t_2\sigma) \cap \overline{\rho}(\phi^\sigma) = \overline{\rho}(t_2) \cap \overline{\rho}(\phi^\sigma)$. Because $\overline{\rho}(\phi^\sigma) = M$, we have $\overline{\rho}(t_1) = \overline{\rho}(t_1\sigma)$ ($\spadesuit$) and $\overline{\rho}(t_2\sigma) = \overline{\rho}(t_2)$ ($\clubsuit$). 

Since $\sigma$ is a unifier, then $t_1\sigma$ and $t_2\sigma$ are syntactical equal. This implies $\overline{\rho}(t_1\sigma) = \overline{\rho}(t_2\sigma)$; by ($\clubsuit$) and ($\spadesuit$) we obtain $\overline{\rho}(t_1) = \overline{\rho}(t_2\sigma) = \overline{\rho}(t_2\sigma) = \overline{\rho}(t_2)$. \qed
\end{proof}

\begin{lemma}
\label{th:equiv}
If $\sigma$ is the mgu of $t_1$ and $t_2$ then $F \models (t_1 = t_2) \leftrightarrow \phi^\sigma$.
\end{lemma}
\begin{proof}
Consequence of Lemmas~\ref{lem:unif} and~\ref{lem:unifback}. \qed
\end{proof}


Now we are ready establish the main contribution of this section, namely that the syntactic unification algorithm is \emph{sound} for semantic unification in ML:

\begin{theorem}[Soundness]
\label{th:main}
Let $\sigma$ be the most general unifier of $t_1$ and $t_2$ obtained by applying the algorithm shown in Figure~\ref{fig:unif} to the unification problem $\{t_1 \eq t_2\}$.
Then $F \models t_1 \land t_2 = t_1 \land \phi^\sigma$ and $F \models t_1 \land t_2 = t_2 \land \phi^\sigma$.
\end{theorem}

\begin{proof}
We have $F \models t_1 \land t_2 = t_1 \land (t_1 = t_2)$ by Proposition~\ref{prop:fpatt}, and $F \models (t_1 = t_2) \leftrightarrow \phi^\sigma$ by Lemma~\ref{th:equiv}. Therefore, $F \models t_1 \land t_2 = t_1 \land \phi^\sigma$. The second conclusion is obtained in a similar way.
\end{proof}
%

Theorem~\ref{th:main} states that if the unification algorithm successfully terminates, then the most general unifier supplies the constraint pattern needed to express the semantic unifier as a conjunction of a structural pattern and a constraint.


\subsubsection{Completeness.}
\label{sec:completeness}
An interesting question to ask here is what happens when the input term patterns are not unifiable? 
In such a case, the unification algorithm fails and the sequence of transformations over the term patterns ends with $\pmb{\bot}$.
In fact, the condition $P'\neq \pmb{\bot}$ in Lemma~\ref{lem:sim} prevents exactly this situation to happen. 
In order to remove this condition, one needs to prove $\phi^P \Rightarrow \bot$ when {\bf Occurs check} and {\bf Symbol clash} apply. 

The injectivity axiom is not enough to prove these properties and stronger axioms are needed.
Obviously, a tempting alternative is to use  constructors instead of injections.
In~\cite{rosu-2017-lmcs}, the \emph{constructors} are defined as follows:\vspace*{-1ex}
\begin{itemize}
\small
\item {\bf{No junk}}: \small$F \models {\bf{\bigvee}}_s {\exists x_1:s_1\ldots\exists x_m:s_m \,.\, c(x_1, \ldots, x_m)}$, where $c \in \Sigma_{s_1 \ldots s_m, s}$;
\item {\bf{No confusion, different constructors}}:\\\small$F \models \lnot(c(x_1, \ldots, x_m) \land c'(y_1, \ldots, y_n))$, with $c \!\neq\!c'$, $c\!\in\! \Sigma_{s_1 .. s_m, s}$, and $c \!\in\! \Sigma_{s_1 .. s_n, s}$.
\item {\bf{No confusion, same constructors}}:\\\small$F \models c(x_1,\ldots,x_m) \land c(y_1,\ldots,y_m) \rightarrow c(x_1 \land y_1,\ldots,x_m \land y_m)$, with $c \in \Sigma_{s_1 \ldots s_m, s}$.\vspace*{-1ex}
\end{itemize}
{\bf No junk} ensures the that constructors can be used to construct all the elements of the target domain. 
{\bf{No confusion, different constructors}} ensures that constructors yield a unique way to construct each element of the target domain. 
{\bf{No confusion, same constructors}} says that constructors are injective.

The {\bf{no confusion, different constructors}} axiom is sufficient to prove Lemma~\ref{lem:sim} for the {\bf Symbol clash} case. Unfortunately, none of these axioms is enough to prove the lemma for the {\bf Occurs check} case. The main issue is that ${x\!= \!f(t_1,.., t_n)}$ cannot be proved equal to $\bot$ when $x \!\in\! \var(f(t_1,.., t_n))$. Recall that the condition $x\! \in \!\var(f(t_1,.., t_n))$ implies that $x$ occurs at least in a~term~$t_i$. 

The axioms for constructors cannot prevent to have $M\models x = f(t_1,.., t_n)$ for some $(S,\Sigma, F)$-model $M$, when $x \in \var(f(t_1,.., t_n))$. Here is a counterexample. 

%
Let $s$ be a sort and and $\Sigma$ a signature which includes only a functional symbol $f\in \Sigma_{s,s}$.
Also, let $M$ be a ML model where $f_M(a) = a$, with $a$ the only element in $M$. 
Note that any valuation $\rho: \Var \to M$, assigns to variables a set equal to $\{a\}$.
Also, note that $f$ satisfies the axioms above: first, the {\bf{no confusion, different constructors}} is satisfied trivially since there is no other symbol in $\Sigma$; second, the {\bf{no confusion, same constructors}} holds, since $\overline{\rho}(f(x) \land f(y)) = \overline{\rho}(f(x)) \cap \overline{\rho}(f(y)) = \{ a \} \cap \{ a \} = \{ a \} = \{ f_M(a) \} \subseteq \overline{\rho}(f(x \land y))$; finally, the {\bf no junk} axiom $\exists x . f(x)$ holds, since $\overline{\rho}(\exists x . f(x)) = \bigcup_{a \in M}\overline{\rho[a/x]}(f(x)) = \bigcup_{a \in M}f_M(a) = M$. However, $x$ and $f(x)$ are unifiable in the sense of ML. 

In our opinion, there are two choices to handle such situations. First, we can modify the syntactic unification algorithm such that it reports also the mappings $x\mapsto t(x)$ when {\bf Occurs check} is applicable (here, $t(x)$ denotes a term that has $x$ as subterm).  If we want to consider only models where the equalities $x=t(x)$ do not hold, then we simply add the axioms $\lnot (x = t(x))$ to $F$. The problem here is that we do not know a priori these axiom patterns. Second, if we want to consider models where  the equalities $x=t(x)$ may hold, then we define $\phi^{t_1=t_2}$ as being $(\bigwedge x=t(x))\implies \phi^\sigma$, where $(\bigwedge x=t(x))$ is the conjunction over all mappings introduced by {\bf Occurs check}, and $\sigma$ is the substitution defined by the {\bf Elimination} mappings. The price paid in this case is that we may get formulas that SMT solvers might not be able to handle.



\section{Generating proofs}
\label{sec:proofs}
In this section we present a sound strategy to generate formal proofs of equivalence between $t_1 \land t_2$ and $t_1 \land \phi^\sigma$ with $\sigma$ the most general unifier of $t_1$ and $t_2$. This strategy uses the rules of the ML proof system~\cite{rosu-2017-lmcs} and some derived rules that mimic the steps of the unification algorithm.

We first explain the main idea of our strategy using Example~\ref{ex:mgu}. 
The equations (\ref{a:start}-\ref{a:elim2}) correspond to the steps of the unification algorithm shown in Figure~\ref{fig:unif}: {\bf Decomposition} for equations (\ref{a:dec1},\ref{a:dec2},\ref{a:dec3}), {\bf Orient} for (\ref{a:orient}), and {\bf Elimination} for (\ref{a:elim1},\ref{a:elim2}). The only exception is the equation (\ref{a:fpatt}), which is justified by Proposition~\ref{prop:fpatt}.
This particular example suggests that successive transformations over the initial pattern produce a conjunction of a term pattern and a constraint. 
The fact that in ML we can express the mgu of term patterns $t_1$ and $t_2$ as a ML pattern $t_1\land t_2$ (i.e., at the \emph{object level}) is important: 
 using the transformations above we can actually generate a \emph{proof certificate} that the obtained constrained term pattern $t_1 \land \phi^\sigma$ is equal  to $t_1 \land t_2$. 

Our current approach is to generate proofs in two stages: first, we start with $t_1 \land t_2$ and we derive $t_1 \land \phi^\sigma$ using several derived proof rules which mimic the steps of the unification algorithm; these will be we proved separately using the ML proof system
; 
second, we start with $t_1 \land \phi^\sigma$ and we derive $t_1 \land t_2$ using the original proof system of ML. For both stages we have strategies that always produce proofs when the most general unifier exists.

\subsubsection{Stage 1}The list of derived rules that we use in the first stage is shown below. For each rule we indicate the  corresponding rule from the unification algorithm:
\newcounter{derived}
\newcommand{\derivedrule}[1]{\refstepcounter{derived}\label{#1}$\Delta$\arabic{derived}.}
\renewcommand{\thederived}{$\Delta$\arabic{derived}}

\begin{center}
\small
\begin{tabular}{l l l}
\derivedrule{delete}& $F \vdash \varphi \land (t = t) \rightarrow \varphi$ & ~~{\small \bf \!\!Delete} \\
\derivedrule{decomposition}& $F \vdash \varphi \land (f(t_1,.., t_n) = f(t'_1,.., t'_n)) \rightarrow \varphi \land t_1 = t'_1 \land..\land t_n = t'_n$ & ~~{\small \bf \!\!Decomposition}\\
\derivedrule{orient}& $F \vdash \varphi \land (f(t_1,.., t_n) = x) \rightarrow  \varphi \land (x = f(t_1,.., t_n))$ & ~~{\small \bf \!\!Orient} \\
\derivedrule{elimination} & $F \vdash \varphi \land (x = t) \rightarrow \varphi[t/x] \land (x = t)$, if $x\not\in \var(t), x \in \var(\varphi)$ & ~~{\small \bf \!\!Elimination}
\end{tabular}
\end{center}
\noindent
These rules are proved (semantically) in the proof of Lemma~\ref{lem:sim}, but we also prove them using the ML proof system (Table~\ref{table:derivedproofs}). Note that there are no corresponding rules for {\bf Occurs check} and {\bf Symbol clash}, because we are interested in generating proofs only for the cases when the most general unifier exists.
An example of a proof that uses the derived rules is shown below:
\newcounter{cellcntr}
\newcommand{\proofline}[1]{\refstepcounter{cellcntr}\label{#1}\roman{cellcntr}~}
\renewcommand{\thecellcntr}{\roman{cellcntr}}
\newcommand{\assumption}{\emph{hypothesis}}
\newcommand{\premise}{\emph{premise}}
\begin{center}
\small
\begin{tabular}{llcl}
\proofline{passumption11}  & $ f(x, g(1), g(z)) \land f(g(y), g(y), g(g(x))$ && \assumption\\
\proofline{pprop.524} & $f(x, g(1), g(z)) \land (f(x, g(1), g(z)) = f(g(y), g(y), g(g(x))))$ && \emph{Prop~\ref{prop:fpatt}}: \ref{passumption11} \\
\proofline{pdec1} & $f(x, g(1), g(z)) \land (x \!=\! g(y)) \land (g(1) = g(y)) \land (g(z) \!=\! g(g(x)))$ && \ref{decomposition}: \ref{pprop.524} \\
\proofline{pdec2} & $f(x, g(1), g(z)) \land (x = g(y)) \land (1 = g(y)) \land (g(z) = g(g(x)))$ && \ref{decomposition}: \ref{pdec1} \\
\proofline{pdec3} & $f(x, g(1), g(z)) \land (x = g(y)) \land (1 = y) \land (z = g(x))$ && \ref{decomposition}: \ref{pdec2} \\
\proofline{pornt} & $f(x, g(1), g(z)) \land (x = g(y)) \land (y = 1) \land (z = g(x))$ && \ref{orient}: \ref{pdec3} \\
\proofline{pelim1} & $f(x, g(1), g(z)) \land (x = g(1)) \land (y = 1) \land (z = g(x))$ && \ref{elimination}:  \ref{pornt} \\
\proofline{pelim2} & $f(x, g(1), g(z)) \land (x = g(1)) \land (y = 1) \land (z = g(g(1)))$ && \ref{elimination}:\ref{pelim1} \\
\end{tabular}
\end{center}



\noindent
Each line represents a proof step annotated with a justification specified as $\langle$the applied proof rule$\rangle$:$\langle$references to previous steps$\rangle$. We intentionally omit $F \vdash$ before each proof step and we prefer to add some useful annotations at the end. 

The first line is our hypothesis. 
The pattern derived at the second line is obtained by applying Proposition~\ref{prop:fpatt} to pattern \ref{passumption11}.
Then, the strategy is given by the unification algorithm.
The third line is obtained by applying \ref{decomposition} to \ref{pprop.524}, that is, {\bf Decomposition} for symbol $f$. 
To keep the above proof simple, we silently use the associativity and commutativity of $\land$.
Starting with the formula at step \ref{passumption11} we are able to derive the formula from step \ref{pelim2}. 

\begin{table}
\label{table:derivedproofs}
\begin{center}
\begin{tabular}{|llcl|}
\hline
\setcounter{cellcntr}{0}
\noindent
\ref{delete} &(\emph{Delete}):&&\\
\proofline{del}  & $\varphi \land t=t$ && \assumption\\
\proofline{dellhs} & $\varphi$ && \ref{propositional}: \ref{dec}\\
\hline

\setcounter{cellcntr}{0}
\noindent
\ref{decomposition} & (\emph{Decomposition}):&&\\
\proofline{dec}  & $\varphi \land (f(t_1,.., t_n) = f(t'_1,.., t'_n))$ && \assumption\\
\proofline{declhs} & $\varphi$ && \ref{propositional}: \ref{dec}\\
\proofline{decrhs} & $f(t_1,.., t_n) = f(t'_1,.., t'_n)$ && \ref{propositional}: \ref{dec}\\
\proofline{decinj} & $f(t_1,.., t_n) = f(t'_1,.., t'_n) \rightarrow t_1 = t'_1 \land..\land t_n = t'_n$ && \emph{inj axiom} in F\\
\proofline{decmp} &$t_1 = t'_1 \land..\land t_n = t'_n$ && \ref{modusponens}: \ref{decrhs}, \ref{decinj}\\
\proofline{decfin} &$\varphi \land t_1 = t'_1 \land..\land t_n = t'_n$&& \ref{propositional}: \ref{declhs}, \ref{decmp}\\

\hline
\setcounter{cellcntr}{0}
\noindent
\ref{orient} &(\emph{Orient}):&&\\
\proofline{ori}  & $\varphi \land (f(t_1,.., t_n) = x)$ && \assumption\\
\proofline{orilhs} & $\varphi$ && \ref{propositional}: \ref{dec}\\
\proofline{orirhs} & $f(t_1,.., t_n) = x$ && \ref{propositional}: \ref{dec}\\
\proofline{orisym} & $x = f(t_1,.., t_n)$ && (symmetry of =): \ref{orirhs}\\
\proofline{orifinal}  & $\varphi \land (x = f(t_1,.., t_n) )$ && \ref{propositional}: \ref{orilhs}, \ref{orisym}\\
\hline

\setcounter{cellcntr}{0}
\noindent
\ref{elimination} &(\emph{Elimination}):&&\\
\proofline{elim}  & $\varphi \land (x=t)$ && \assumption\\
\proofline{elimlhs} & $\varphi$ && \ref{propositional}: \ref{elim}\\
\proofline{elimrhs} & $x = t$ && \ref{propositional}: \ref{elim}\\
\proofline{elimintro} &  $\varphi[x/x]$ && \ref{elimlhs}: ($\varphi = \varphi[x/x]$)\\
\proofline{elimint} & $x = t \land \varphi[x/x]$ && \ref{propositional}: \ref{elimrhs}, \ref{elimintro}\\
\proofline{elimsub} & $(x = t) \land \varphi[x/x] \rightarrow \varphi[t/x] $ && \ref{equalityelimination}\\

\proofline{elimsub} & $\varphi[t/x]$ && \ref{modusponens}: \ref{elimint}, \ref{elimsub}\\
\proofline{elimfinal} & $\varphi[t/x] \land x \!=\! t$ && \ref{propositional}: \ref{elimsub}, \ref{elimrhs}\\
\hline
\setcounter{cellcntr}{0}
\noindent
Prop\,\ref{prop:fpatt}\label{p3to}&($\leftarrow$):&&\\
\proofline{p2}  & $\varphi \land (\varphi = \varphi')$ && \assumption\\
\proofline{p21} & $\varphi = \varphi'$ && \ref{propositional}: \ref{p2} \\
\proofline{p121} & $\varphi$ && \ref{propositional}: \ref{p2} \\
\proofline{p24} & $\varphi'$&& \ref{propositional}: \ref{p21}, \ref{p121}, Prop. \ref{prop:eq}\\
\proofline{p25}  & $\varphi \land \varphi'$ && \ref{propositional}: \ref{p121}, \ref{p24}\\
\proofline{p26}  & $(\varphi \land \varphi') \rightarrow (\varphi \land (\varphi = \varphi')\rightarrow(\varphi \land \varphi'))$ && \ref{propositional}: \ref{l1}\\
\proofline{p27}  & $\varphi \land (\varphi = \varphi')\rightarrow (\varphi \land \varphi')$ && \ref{modusponens}: \ref{p25}, \ref{p26}\\
\proofline{p28}  & $ (\varphi \land \varphi')$ && \ref{modusponens}: \ref{p2}, \ref{p27}\\
\hline
\setcounter{cellcntr}{0}
\noindent
Prop\,\ref{prop:fpatt}\label{p3from} &($\rightarrow$):&&\\
\proofline{p34}  & $ (\varphi \land \varphi') $ && \assumption\\
\proofline{p35}  & $ \dfness{\varphi \land \varphi'}{}{}$ && \emph{definedness axiom in $F$} \\
\proofline{p361}  & $  \dfness{\varphi \land \varphi'}{}{} \rightarrow ((\varphi \land \varphi') \rightarrow \dfness{\varphi \land \varphi'}{}{})$ && \ref{propositional} (\ref{l1})\\
\proofline{p36}  & $ (\varphi \land \varphi') \rightarrow \dfness{\varphi \land \varphi'}{}{}$ && \ref{modusponens}: \ref{p35}, \ref{p361}\\
\proofline{p37}  & $ (\varphi \land \varphi') \rightarrow \varphi \in \varphi'$ && \emph{definition of $\in$}: \ref{p36}\\
\proofline{p38}  & $ (\varphi \land \varphi') \rightarrow (\varphi = \varphi')$ && \ref{membershipequality}: \ref{p37}\\
\proofline{p40}  & $ (\varphi \land \varphi') \rightarrow \varphi$ && \ref{propositional}$^+$\\
\proofline{p39}  & $ (\varphi \land \varphi') \rightarrow \varphi \land (\varphi = \varphi')$ && \ref{propositional}$^+$: \ref{p38}, \ref{p40}\\
\proofline{p50}  & $ \varphi \land (\varphi = \varphi')$ && \ref{modusponens}: \ref{p34}, \ref{p39}\\
\hline
\end{tabular}
\end{center}
\caption{Proofs of the derived rules \ref{delete}-\ref{elimination} and Proposition~\ref{prop:fpatt}.}
\end{table}

It is easy to see that the strategy of the first stage is dictated by the unification algorithm shown in Figure~\ref{fig:unif}. Its soundness is given by Proposition~\ref{prop:fpatt} and Lemma~\ref{lem:sim}.
However, we provide proofs that use the rules of the ML proof system for \ref{delete}-\ref{elimination} and Proposition~\ref{prop:fpatt} in Table~\ref{table:derivedproofs} .
It is worth noting that we used only a few rules of the ML proof system: \ref{propositional}, \ref{modusponens}, \ref{equalityelimination}, \ref{membershipequality}.

To validate the proofs from Table~\ref{table:derivedproofs}, we have encoded the definitions, the proof rules and the needed axioms in Coq. Then we checked our proofs mechanically. 
We have formulated and proved a deduction theorem which holds for the ML proof system fragment that we use (\ref{propositional}, \ref{modusponens}, \ref{equalityelimination}, \ref{membershipequality}). Also, we provide Coq proofs for rather trivial steps (i.e., $\land$ elimination, $\land$ introduction, and other simple propositional lemmas) using the rules in Figure~\ref{fig:prop}. The (assertive) proof style that we used is intended to improve source code reading for non-expert Coq users. The Coq code can be found at~\cite{coqfile}. 

\subsubsection{Stage 2} We start explaining our strategy for stage 2 by proving the reversed implication of our example:

\begin{center}
\small
\begin{tabular}{llcl}
\proofline{assumption1}  & $f(x, g(1), g(z)) \land (x = g(1)) \land (y = 1) \land (z = g(g(1)))$ && \assumption\\
\proofline{lhs0} & $ f(x, g(1), g(z)) $ && \ref{propositional}: \ref{assumption1}\\
\proofline{rhs1}  & $(x = g(1)) \land (y = 1) \land (z = g(g(1)))$ && \ref{propositional}: \ref{assumption1}\\
\proofline{lhs1}  & $x = g(1)$  &&  \ref{propositional}: \ref{rhs1}\\
\proofline{rhs2} & $(y = 1) \land z = g(g(1))$ &&  \ref{propositional}: \ref{rhs1}\\
\proofline{lhs2}  & $y = 1$ &&  \ref{propositional}: \ref{rhs2}\\
\proofline{rhs3}  & $z = g(g(1))$ &&  \ref{propositional}: \ref{rhs2}\\
\proofline{refl1} & $f(x, g(1), g(z)) = f(x, g(1), g(z))$ && \ref{equalityintroduction}\\
\proofline{refl2} & $f(g(y), g(y), g(g(x)) = f(g(y), g(y), g(g(x))$ && \ref{equalityintroduction}\\
\proofline{elim1} & $ f(g(1), g(1), g(z)) = f(x, g(1), g(z))$ && \ref{equalityelimination}: \ref{refl1},\ref{lhs1}\\
\proofline{elim2} & $ f(g(1), g(1), g(g(g(1)))) = f(x, g(1), g(z))$ && \ref{equalityelimination}: \ref{elim1}, \ref{rhs3}\\
\proofline{elim3} & $f(g(1), g(y), g(g(x)) = f(g(y), g(y), g(g(x))$ && \ref{equalityelimination}: \ref{refl2},\ref{lhs2}\\
\proofline{elim4} & $f(g(1), g(1), g(g(x)) = f(g(y), g(y), g(g(x))$ && \ref{equalityelimination}: \ref{elim3}, \ref{lhs2}\\
\proofline{elim5} & $f(g(1), g(1), g(g(g(1))) = f(g(y), g(y), g(g(x))$ && \ref{equalityelimination}: \ref{elim4},\ref{lhs1}\\
\proofline{elim6} & $ f(x, g(1), g(z)) = f(g(y), g(y), g(g(x))$ && \ref{equalityelimination}: \ref{elim2}, \ref{elim5}\\
\proofline{prop} & $  f(x, g(1), g(z)) \!\land\! (f(x, g(1), g(z)) = f(g(y), g(y), g(g(x)))$ && \ref{propositional}: \ref{lhs0}, \ref{elim6}\\
\proofline{final} & $  f(x, g(1), g(z)) \land f(g(y), g(y), g(g(x))$ && Prop~\ref{prop:fpatt}: \ref{prop}\\
\end{tabular}
\end{center}

\noindent
Now, we present the strategy corresponding to this stage, which has five steps:

\renewcommand{\theenumi}{\arabic{enumi}}
\begin{enumerate}
\item start with $t_1 \land \phi^\sigma$ as {\assumption};
\item use \ref{propositional} to break the large conjunction from the {\assumption} (e.g., steps \ref{lhs0}-\ref{rhs3})
\item use  \ref{equalityintroduction} to introduce equalities $t_1 = t_1$ and $t_2 = t_2$ (e.g., steps \ref{refl1}, \ref{refl2});
\item use $\ref{equalityelimination}$ to replace the variables occurring in the left hand sides of the equalities (e.g., \ref{elim2}, \ref{elim5});
\item use $\ref{equalityelimination}$ to equate the right hand sides of the equalities produced by the previous step(e.g. \ref{elim6}); then apply $\ref{propositional}$ ($\land$ introduction, e.g., \ref{prop}), and finally Proposition~\ref{prop:fpatt} (e.g., \ref{final}).
\end{enumerate}

This strategy essentially rebuilds the semantic unifier $t_1 \land t_2$ starting with $t_1 \land \phi^\sigma$. Because $\phi^\sigma$ has the form $\bigwedge_{i=1}^n x_i = u_i$ the step 2 will always produce equalities of the form $x_i = u_i$ for all $i=\overline{1,n}$. In the left hand sides of the equalities introduced by step 3 we can always substitute $x_i$ by $u_i$. Since $\sigma$ is the most general unifier, the left hand sides will become equal after substitutions performed by step 4. Finally, we can always apply \ref{equalityelimination}, \ref{propositional}, and Proposition~\ref{prop:fpatt} conveniently  to obtain $t_1 \land t_2$.
Because it uses only rules from the original proof system of ML (check Table~\ref{table:derivedproofs} for proof of Proposition~\ref{prop:fpatt}), this strategy is sound. 

\renewcommand{\theenumi}{\arabic{enumi}}



\section{Conclusions}

Previous verification efforts with ML~\cite{rosu-stefanescu-2012-fm,DBLP:conf/lics/RosuSCM13,DBLP:conf/icse/RosuS11,DBLP:conf/wrla/RusuA16,arusoaie:hal-01627517,DBLP:conf/birthday/LucanuRAN15,stefanescu-park-yuwen-li-rosu-2016-oopsla,moore-pena-rosu-2018-esop,park-zhang-saxena-daian-rosu-2018-fse,stefanescu-ciobaca-mereuta-moore-serbanuta-rosu-2014-rta} were based on unification. However, unification was always considered a trusted component. 

In this paper we finally tackle down this issue by proposing a sound method for unification which involves a syntactic unification algorithm. 
More precisely, we show that the syntactic unification algorithm proposed by Martelli and Montanari~\cite{martelli} is \emph{sound} for semantic unification in ML. We explain by means of a counterexample, why this algorithm is not \emph{complete} for semantic unification.
Finally, we show a \emph{provableness} property of the same algorithm: we provide a sound strategy to generate a proof certificate when the most general unifier exists. This proof uses some derived rules (which we encode and prove in Coq) and the rules of the ML proof system~\cite{rosu-2017-lmcs}.

\paragraph{Related work.}We include here only the comparison with the closest related work Kore~\cite{koredocs}: an implementation of ML which is currently under development~\cite{kore}. 
They handle conjunctions via a set of  transformations over patterns intended to serve a more general purpose, for instance, to deal with partiality and injections (subsort relations). 
The approach that we proposed here focuses on how the syntactic unification algorithms can be used to help reasoning in ML.\vspace*{-1ex}

\paragraph{Future work.}
The fact that the proof of the soundness of our approach depends on the unification algorithm is intriguing. We intend to explore whether there is an independent proof, which uses only the definition of the most general unifier.

A topic that also needs further investigation is the completeness of the algorithm with respect to semantic unification.  In Section~\ref{sec:completeness} we discuss the completeness issue and we sketch two solutions, but a deeper investigation is required.

Via private communication with the Kore team we learned that a slightly modified proof system is implemented in Kore for which a deduction theorem can be proved. We intend to adapt our proof generation strategy to use this new proof system since it seems that the deduction theorem can simplify some steps.

Finally, a completely new ground to explore is unification modulo axioms (e.g., commutativity, associativity, and so on). Obviously, it is more challenging to use the existing unification modulo axioms algorithms in the same manner as we have done for syntactic unification.

\subsubsection{Acknowledgements.}
We would like to especially thank the Kore developers and researchers: Phillip Harris, Traian {\c S}erb{\u a}nu{\c t}{\u a} and Virgil {\c S}erb{\u a}nu{\c t}{\u a} for their valuable assistance and feedback. They helped us with our proof generation strategy and they suggested improvements for our current work.
We also want to specially thank Grigore Ro{\c s}u for the fruitful discussions that we had about this topic at FROM 2018.
This work was supported by a grant of the ``Alexandru Ioan Cuza'' University of Ia{\c s}i, within the Research Grants program, Grant UAIC, code GI-UAIC-2017-08.
%
\bibliographystyle{splncs04}
\bibliography{refs}

\end{document}